%% file: Undecidability_of_PRS.tex
\documentclass{llncs}
\usepackage{etex}

\input{preamble-llncs}

\renewcommand{\act}[2][]{\stackrel{#2}{\longrightarrow}^{#1}}

\usepackage{tikz}

\title{Branching Bisimilarity Checking for PRS}
\author{Qiang Yin, Yuxi Fu, Chaodong He, Mingzhang Huang and Xiuting Tao}
\institute{BASICS, Department of Computer Science, Shanghai Jiao Tong University}
\date{January 7, 2014}

\begin{document}

\maketitle

\begin{abstract}
Recent studies reveal that branching bisimilarity is decidable for both nBPP (normed Basic Parallel Processes) and nBPA (normed Basic Process Algebras).
These results lead to the question if there are any other models in the hierarchy of PRS (Process Rewrite Systems) whose branching bisimilarity is decidable.
It is shown in this paper that the branching bisimilarity for both nOCN (normed One Counter Nets) and nPA (normed Process Algebras) is undecidable.
These results essentially imply that the question has a negative answer.
\end{abstract}

\section{Introduction}

Verification on infinite-state systems has been intensively studied for the past two decades~\cite{BurkartCaucalMollerSteffen2001,KuceraJancar2006}.
One major concern in these studies is equivalence checking.
Given a specification $\mathcal{S}$ of an intended behaviour and a claimed implementation $\mathcal{I}$ of $\mathcal{S}$, one is supposed to demonstrate that $\mathcal{I}$ is correct with respect to $\mathcal{S}$.
A standard interpretation of correctness is that an implementation should be behaviourally equivalent to its specification.
Among all the behavioural equalities studied so far, bisimilarity stands out as the most abstract and the most tractable one.
Two well known bisimilarities are the strong bisimilarity and the weak bisimilarity due to Park~\cite{Park1981} and Milner~\cite{Milner1989}.
Considerable amount of effort has been made to investigate the decidability and the algorithmic aspect of the two bisimilarities on various models of infinite state system~\cite{Srba2002b}.
These models include pushdown automata, process algebras, Petri nets and their restricted and extended variations.
An instructive classification of the models in terms of PRS (Process Rewrite Systems) is given by Mayr~\cite{Mayr2000}.

The strong bisimilarity checking problem has been well studied for PRS hierarchy.
Influential decidability results include for example~\cite{BaetenBergstraKlop1987,ChristensenHuttelStirling1992,ChristensenHirshfeldMoller1993,Stirling1996,HirshfeldJerrum1999}.
On the negative side, Jan\v{c}ar attained in~\cite{Jancar1995} the undecidable result of strong bisimilarity on nPN (normed Petri Nets).
The proof makes use of a powerful technique now known as Defender's Forcing~\cite{JancarSrba2008}, which remains a predominant tool to establish negative results about equivalence checking.

In the weak case the picture is less clear.
It is widely believed that weak bisimilarity is decidable for both nBPA (normed Basic Process Algebras) and nBPP (normed Basic Parallel Processes).
The problem has been open for a long time.
Srba~\cite{Srba2002a} showed that weak bisimilarity on nPDA (normed Pushdown Automata) is undecidable by a reduction from the halting problem of Minsky Machine.
The undecidability was soon extended to nOCN (normed One Counter Nets), a submodel of both nPDA and nPN, by Mayr~\cite{Mayr2003}.
Srba also showed that the weak bisimilarity on PA (Process Algebras) is undecidable~\cite{Srba2003}.
Later several highly undecidable results were established by Jan\v{c}ar and Srba~\cite{Srba2004,JancarSrba2002,JancarSrba2008} for the weak bisimilarity checking problem on PN, PDA and PA.

The decidability of the weak bisimilarity on nBPA and nBPP has been open for well over twenty years.
Encouraging progress has been made recently.
\mbox{Czerwi\'{n}ski}, Hofman and Lasota proved that branching bisimilarity, a standard refinement of the weak bisimilarity, is decidable on nBPP~\cite{Czerwinski2011}.
The novelty of their approach is the discovery of some kind of normal form for nBPP.
Using a quite different technique Fu showed that the branching bisimilarity is also decidable on nBPA~\cite{Fu2013}.
In retrospect one cannot help thinking that more attention should have been paid to the branching bisimilarity.
Going back to the original motivation to equivalence checking, one would agree that a specification $\mathcal{S}$ normally contains no silent actions because silent actions are about how-to-do.
Consequently all the silent actions introduced in an implementation must be bisimulated vacuously by the specification.
It follows that $\mathcal{S}$ is weakly bisimilar to an implementation $\mathcal{I}$ if and only if $\mathcal{S}$ is branching bisimilar to $\mathcal{I}$.
What this observation tells us is that as far as verification is concerned the branching bisimilarity ought to play a role no less than the weak bisimilarity.

\begin{figure*}[t]
    \begin{center}
      \begin{tabular}{|c||c|c|c|c|c|}
        \hline
         & $\ $\textsf{nBPA}$\ $ & $\ $\textsf{nBPP}$\ $ & $\ $\textsf{nPDA}$\ $ & $\ $\textsf{nPA}$\ $ & $\ $\textsf{nPN}$\ $ \\ \hline\hline
        $\ $Strong Bisimilarity$\ $ & $\checkmark$\cite{BaetenBergstraKlop1987} & $\checkmark$\cite{ChristensenHirshfeldMoller1993} & $\checkmark$\cite{Stirling1996} & $\checkmark$\cite{HirshfeldJerrum1999} & $\mathbf{\times}$\cite{Jancar1995} \\ \hline
        $\ $Branching Bisimilarity$\ $ & $\checkmark$\cite{Fu2013} & $\checkmark$\cite{Czerwinski2011} & $\ \mathbf{\times}$[this paper]$\ $ & $\ \mathbf{\times}$[this paper]$\ $ & $\times$\cite{Jancar1995} \\ \hline
        $\ $Weak Bisimilarity$\ $ & {\bf ?} & {\bf ?} & $\mathbf{\times}$\cite{Mayr2003} & $\mathbf{\times}$[this paper] & $\mathbf{\times}$\cite{Jancar1995} \\ \hline
      \end{tabular}
      \caption{Decidability of Branching Bisimilarity for Normed PRS}
      \label{Decidability-of-PRS}
    \end{center}
\end{figure*}

The above discussion suggests to address the following question:
Is there any other model in the PRS hierarchy whose branching bisimilarity is decidable?
The purpose of this paper is to resolve this issue.
Our contributions are as follows:
\begin{itemize}
\item
We establish the fact that on both nOCN and nPA every relation between the branching bisimilarity and the weak bisimilarity is undecidable.
These are improvement of Mayr's result about the undecidability of the weak bisimilarity on nOCN~\cite{Mayr2003} and Srba's result~\cite{Srba2003} about the undecidability of the weak bisimilarity on PA.
These new results together with the previous (un)decidability results about the {\em normed} models in PRS are summarized in Fig.~\ref{Decidability-of-PRS}, where a tick is for `decidable' and a cross for `undecidable'.
\item
We showcase the subtlety of Defender's Forcing technique usable in branching bisimulation game.
It is pointed out that the technique must be of a semantic nature for it to be applicable to the branching bisimilarity.
\end{itemize}
The two negative results imply that in the PRS hierarchy the branching bisimilarity on every normed model above either nBPA or nBPP is undecidable.

The rest of the paper is organized as follows.
Section~\ref{sec:Preliminaries} introduces the necessary preliminaries.
Section~\ref{sec-Defender-Forcing-for-Branching-Bisimulation} establishes the undecidability result for nOCN and demonstrates Defender's Forcing technique for branching bisimulation game.
Section~\ref{sec-Undecidability-of-nPA} proves the undecidability result about nPA.
Section~\ref{sec-Conclusion} concludes.

\section{Preliminaries}\label{sec:Preliminaries}

A {\em process algebra} $\mathcal{P}$ is a triple $(\mathcal{C}, \mathcal{A}, \Delta)$, where $\mathcal{C}$ is a finite set of process constants, $\mathcal{A}$ is a finite set of actions ranged over by $\ell$, and $\Delta$ is a finite set of transition rules.
The {\em processes} defined by $\mathcal{P}$ are generated by the following grammar:
\[
P \  \Coloneqq  \  \epsilon  \  \mid \  X \  \mid \  PP'  \  \mid \   P{\,\|\,}P'.
\]
The grammar equality is denoted by $=$.
We assume that the sequential composition $PP'$ is associative up to $=$ and the parallel composition $P{\,\|\,} P'$ is associative and commutative up to $=$.
We also assume that $\epsilon P=P\epsilon=\epsilon {\,\|\,} P=P {\,\|\,} \epsilon=P$.
There is a special symbol $\tau$ in $\mathcal{A}$ for silent transition.
The set $\mathcal{A}\setminus\{\tau\}$ is ranged over by $a,b,c,d$.
The transition rules in $\Delta$ are of the form $X \act{\ell} P$.
The following labeled transition rules define the operational semantics of the processes.
\[\begin{tabular}{cccc}
  $\inference{X\act{\ell}P\in\Delta}{X\act{\ell}P}$ &
  $\inference{P\act{\ell}P'}{PQ\act{\ell}P'Q}$ &
  $\inference{P\act{\ell}P'}{P{\,\|\,} Q\act{\ell}P'{\,\|\,} Q}$ & $\inference{Q\act{\ell}Q'}{P{\,\|\,} Q\act{\ell}P{\,\|\,} Q'}$\\
\end{tabular}\]
The operational semantics is structural, meaning that $PQ\act{\ell}P'Q$, $P{\,\|\,} Q\act{\ell}P'{\,\|\,} Q$ and $Q{\,\|\,} P\act{\ell}Q{\,\|\,} P'$ whenever $P\act{\ell}P'$.
We write $\Longrightarrow $ for the reflexive transitive closure of $\stackrel{\tau}{\longrightarrow}$, and $\stackrel{\widehat{\ell}}{\Longrightarrow}$ for $\Longrightarrow\stackrel{\ell}{\longrightarrow}\Longrightarrow$ if $\ell\ne\tau$ and for $\Longrightarrow$ otherwise.

A {\em one counter net} $\mathcal{M}$ is a 4-tuple $(\mathcal{Q}, X, \mathcal{A}, \Delta)$, where $\mathcal{Q}$ is a finite set of states ranged over by $p,q,r,s$, $X$ represents a place, $\mathcal{A}$ is a finite set of actions as in a process algebra, and $\Delta$ is a finite set of transition rules.
A {\em process} defined by $\mathcal{M}$ is of the form $pX^{n}$, where $n$ indicates the number of tokens in $X$.
A transition rule in $\Delta$ is of the form $pX^{i}\stackrel{\ell}{\longrightarrow}qX^{j}$ with $i<2$.
The semantics is structural in the sense that $pX^{i+k}\stackrel{\ell}{\longrightarrow}qX^{j+k}$ whenever $pX^{i}\stackrel{\ell}{\longrightarrow}qX^{j}$.
A process $P$ defined in $\mathcal{P}$, respectively $\mathcal{M}$, is {\em normed} if $\exists \ell_1, \dots, \ell_n.P \act{\ell_1 } \dots  \act{\ell_{n}}\epsilon$, respectively $\exists \ell_1, \dots, \ell_n,p.(P \act{\ell_1 } \dots  \act{\ell_{n}}p)\wedge\forall\ell,Q.\neg(p\stackrel{\ell}{\longrightarrow}Q)$.
We say that $\mathcal{P}$/$\mathcal{M}$ is normed if only normed processes are definable in it.
We write (n)PA for the (normed) Process algebras and (n)OCN for the (normed) One Counter Nets.

In the presence of silent actions two well known process equalities are the weak bisimilarity~\cite{Milner1989} and the branching bisimilarity~\cite{GlabbeekW96}.

\begin{definition}
A relation $\mathcal{R}$ is a {\em weak bisimulation} if the following are valid: \\
1. Whenever $P\mathcal{R}Q$ and $P \act{\ell} P'$, then $Q\stackrel{\widehat{\ell}}{\Longrightarrow}Q'$ and $P'\mathcal{R}Q'$ for some $Q'$. \\
2. Whenever $P\mathcal{R}Q$ and $Q\act{\ell}Q'$, then $P\stackrel{\widehat{\ell}}{\Longrightarrow}P'$ and $P' \mathcal{R} Q'$ for some $P'$. \\
The {\em weak bisimilarity} $\weq$ is the largest weak bisimulation.
\end{definition}

\begin{definition}\label{def:beq}
A relation $\mathcal{R}$ is a {\em branching bisimulation} if the following hold: \\
1. Whenever $P\mathcal{R}Q$ and $P \act{\ell} P'$, then either (i)
    $Q\Longrightarrow Q''\act{\ell}Q'$ and $P'\mathcal{R}Q'$ and $P\mathcal{R}Q''$ for some $Q',Q''$ or (ii) $\ell = \tau$ and $P'\mathcal{R}Q$. \\
2. Whenever $P\mathcal{R}Q$ and $Q\act{\ell}Q'$, then either (i) $P\Longrightarrow P''\act{\ell}P'$ and $P'\mathcal{R} Q'$ and $P''\mathcal{R} Q$ for some $P',P''$ or (ii) $\ell=\tau$ and
    $P\mathcal{R}Q'$. \\
The {\em branching bisimilarity} $\beq$ is the largest branching bisimulation.
\end{definition}

The following lemma, first noticed by van Glabbeek and Weijland~\cite{GlabbeekW96}, plays a fundamental role in the study of branching bisimilarity.

\begin{lemma}\label{computation-lemma}
If $P \Longrightarrow P'\Longrightarrow  P'' \simeq P$ then $P'\simeq P$.
\end{lemma}

Let $\approxeq$ be a process equivalence.
A silent action $P\stackrel{\tau}{\longrightarrow}P'$ is {\em state preserving} with regards to $\approxeq$, notation $P\rightarrow P'$, if $P'\approxeq P$;
it is {\em change-of-state} with regards to $\approxeq$, notation $P\stackrel{\iota}{\longrightarrow}P'$, if $P'\not\approxeq P$.
The reflexive and transitive closure of $\rightarrow$ is denoted by $\rightarrow^{*}$.
Branching bisimilarity strictly refines weak bisimilarity in the sense that only state preserving silent actions can be ignored; a change-of-state must be explicitly bisimulated.
Suppose that $P \beq Q$ and $P \act{\ell} P'$ is matched by the transition sequence $Q \stackrel{\tau}{\longrightarrow} \dots \stackrel{\tau}{\longrightarrow} Q_{i} \stackrel{\tau}{\longrightarrow} \dots \stackrel{\tau}{\longrightarrow}Q'' \act{\ell} Q'$.
By definition one has $P \beq Q''$.
It follows from Lemma~\ref{computation-lemma} that $P \beq Q_{i}$, meaning that all silent actions in $Q\Longrightarrow Q''$ are necessarily state preserving.
This property fails for the weak bisimilarity as the following example demonstrates.
\begin{example}\label{example}
Consider the transition system $\{P\act{b}\epsilon,\ P\stackrel{\tau}{\longrightarrow} P' \act{a} \epsilon,\ P\act{a} \epsilon;\ Q\act{b} \epsilon,\ Q\stackrel{\tau}{\longrightarrow} Q' \act{a} \epsilon\}$.
One has $P\weq Q$.
However $P\not\beq Q$ since $Q\not\simeq Q'$.
\end{example}

A game theoretic characterization of bisimilarity is by~\emph{bisimulation game}~\cite{Stirling1998}.
Suppose that a pair of processes $P,Q$, called a {\em configuration}, are defined in say a process algebra $(\mathcal{C},\mathcal{A},\Delta)$.
A \emph{branching bisimulation game} for the configuration $(P,Q)$ is played between {\em Attacker} and {\em Defender}.
The game is played in {\em rounds}.
A new configuration is chosen after each round.
Every round consists of three steps defined as follows, assuming $(P_{0},P_{1})$ is the current configuration:
\begin{enumerate}
\item
Attacker chooses $i\in\{0,1\}$, $\ell \in \mathcal{A}$ and some process $P_i'$ such that $P_i \stackrel{\ell}{\longrightarrow} P_i'$.
\item
Defender may respond in either of the following manner:
\begin{itemize}
\item Choose some $P_{1-i}', P_{1-i}''$ such that $P_{1-i} \Longrightarrow  P_{1-i}'' \stackrel{\ell}{\longrightarrow} P_{1-i}'$.
\item Do nothing in the case that $\ell=\tau$.
 \end{itemize}
\item
Attacker decides which of $(P_i, P_{1-i}'')$, $(P_i', P_{1-i}')$ is the new configuration if Defender has played.
Otherwise the new configuration must be $(P_i', P_{1-i})$.
\end{enumerate}
In a {\em weak bisimulation game} a round consists of two steps.
The first step is the same as above.
In the second step Defender chooses some $P_{1-i}'$ and some transition sequence $P_{1-i} \stackrel{\widehat{\ell}}{\Longrightarrow} P_{1-i}'$.
The game then continues with $(P_i',P_{1-i}')$.

Defender wins a game if it never gets stuck; otherwise Attacker wins.
We say that Defender/Attacker has a {\em winning strategy} if it can always win no matter how the opponent plays.
The following lemma is well known, a clever use of which often simplifies bisimulation argument considerably.

\begin{lemma}\label{lm:game1}
Defender has a winning strategy in the branching, respectively weak, bisimulation game starting from the configuration $(P,Q)$ if and only if $P\simeq Q$, respectively $P\approx Q$.
\end{lemma}

Attacker has a winning strategy for the branching bisimulation game of the pair $P,Q$ defined in Example~\ref{example}.
It simply chooses $P\act{a} \epsilon$.
If Defender chooses $Q\stackrel{\tau}{\longrightarrow} Q' \act{a} \epsilon$, Attacker chooses the configuration $(P,Q')$ and wins.
Defender can win the weak bisimulation game of $(P,Q)$ though.

\section{Defender's Forcing with Delayed Justification}\label{sec-Defender-Forcing-for-Branching-Bisimulation}

A powerful technique for proving lower bounds for bisimilarity checking problem is Defender's Forcing described by Jan\v{c}ar and Srba in~\cite{JancarSrba2008}.
The basic idea is to force Attacker to make a particular choice in a bisimulation game by introducing enough copycat rules.
An application of the technique to weak bisimulation game should be careful since both Attacker and Defender can take advantage of silent transitions.
The design of a branching bisimulation game is even more subtle.
In such a game a sequence of silent transitions used by Defender, except possibly the last one, must all be state preserving.
A useful technique, motivated by Lemma~\ref{computation-lemma}, is to make use of generating processes.
The process $G$ defined by the rules $G\stackrel{\tau}{\longrightarrow}GX$ and $GX\stackrel{\tau}{\longrightarrow}G$ is {\em generating} due to the fact that every process that $G$ may evolve into, say $GX^{n}$, is branching bisimilar to $G$.
The presence of other transition rules for $G$ and $X$ would not change the fact that $G\simeq GX^{n}$ for all $n$.
This technique has already been used in the design of weak bisimulation games~\cite{JancarSrba2008,Mayr2003}.
The relations these games give rise to are not branching bisimulation because a state-preserving transition may be simulated by a change-of-state silent transition.
In what follows we use a small example to expose the subtlety of branching bisimulation game and the technique to apply Defender's Forcing in such a game.

Mayr proved in~\cite{Mayr2003} a general result that the weak bisimilarity is undecidable for any model that subsumes nOCN.
The lower bound is achieved by reducing from the halting problem of Minsky machine.
A Minsky machine $\mathcal{M}$ with two counters $c_1,c_2$ is a program of the form $1: I_1;\ 2: I_2;\ \dots;\ m{-}1: I_{m-1};\ m: \textrm{halt}$,
where for each $i\in\{1,\dots,m-1\}$ the instruction $I_i$ is in either of the following forms, assuming $j,k\in\{1,\dots,m-1\}$ and $e\in \{1,2\}$,
\begin{itemize}
\item $c_{e}:= c_{e}+1$ and then goto $j$.
\item if $c_{e} = 0$ then goto $j$; otherwise $c_{e}:= c_{e}-1$ and then goto $k$.
\end{itemize}
By encoding a pair of numbers $(n_1, n_2)$ by G\"odel number of the form $2^{n_1}3^{n_2}$, Mayr implemented the increment and decrement operations on the counters by multiplying and dividing by $2$ and $3$ respectively.
The central part of Mayr's proof is to show that it is possible to encode these operations and test for divisibility by constant into weak bisimulation games on nOCN.
We shall show that Mayr's reduction can be strengthened to produce reductions to branching bisimulation games on nOCN.
For every instruction ``$i:I_i$'' of a Minsky machine $\mathcal{M}$ a pair of states $p_i,p_i'$ are introduced.
Suppose ``$i: c_2:= c_2+1; \textrm{ goto } j$'' is the i-th instruction of $\mathcal{M}$.
The instruction is translated to the rules given in Fig.~\ref{fig:ocn_add}.
The model defined in Fig.~\ref{fig:ocn_add} is open-ended.
Transition rules associated to $p_{j}$ and $p_{j}'$ are not given.
We have however the following interesting property.

\begin{figure*}[t]
\begin{center}
\begin{tabular}{|r|l|}
\hline
$p_i \stackrel{\tau}{\longrightarrow}G'\ $ & $\ p_i' \stackrel{\tau}{\longrightarrow}G'$ \\
$p_i \act{a} q_1\ $ & $\ G' \act{a} q_1',\ G'\stackrel{\tau}{\longrightarrow}G'X,\ G'X\stackrel{\tau}{\longrightarrow} G'\ $ \\
\hline
$q_1\act{a}q_2\ $ & $\ q_1'\act{a} q_2'$ \\
$q_1\act{t}t_{3}\ $ & $\ q_1'\act{t} t_{1}$ \\
\hline
$q_2 \stackrel{\tau}{\longrightarrow}G\ $ & $\ q_2' \stackrel{\tau}{\longrightarrow}G$ \\
$\ G\stackrel{\tau}{\longrightarrow}GX,\ GX\stackrel{\tau}{\longrightarrow} G,\ G\act{a} q_3\ $ & $\ q_2' \act{a} q_3'$ \\
\hline
$q_3\act{a}p_j\ $ & $\ q_3'\act{a} p_j'$ \\
$q_3\act{t}t_{1}\ $ & $\ q_3'\act{t} t_{1}$ \\
\hline
$\ t_{3}X\act{c}t''X,\ t''X\act{c}t'X,\ t'X\act{c}t_{3}\ $ & $\ t_{1}X\act{c}t_{1}\ $ \\
\hline
\end{tabular}
\end{center}
\caption{Multiplication Operation on Counter in OCN}
\label{fig:ocn_add}
\end{figure*}

\begin{lemma}\label{2014-01-19}
Let $n=2^{n_1}3^{n_2}$ for some $n_{1},n_{2}$.
Defender of the branching bisimulation game of $(p_jX^{3n}, p_j'X^{3n})$ has a winning strategy if and only if Defender of the branching bisimulation game of $(p_iX^{n}, p_i'X^n)$ has a winning strategy.
\end{lemma}
\begin{proof}
The crucial point here is that the copycat rules $p_i \stackrel{\tau}{\longrightarrow}G'$ and $p_i' \stackrel{\tau}{\longrightarrow}G'$, which syntactically identify what $p_{i}X^{n}$ and $p_{i}'X^{n}$ may reach in one silent step, do not automatically create a Defender's Forcing situation.
The reason is that although $p_{i}'X^{n}\rightarrow G'X^{n}$, since $p_i'X^{n}\stackrel{\tau}{\longrightarrow}G'X^{n}$ is the only action of $p_i'X^{n}$, it might well be that $p_{i}X^{n}\stackrel{\iota}{\longrightarrow} G'X^{n}$.
For branching bisimulation syntactical Defender's Forcing is insufficient.
One needs Defender's Forcing that works at semantic level.
Let's take a look at the development of the game in some detail.
\begin{enumerate}
\item
If Attacker plays $p_{i}X^{n}\stackrel{\tau}{\longrightarrow}G'X^{n}$, Defender plays $p_{i}'X^{n}\stackrel{\tau}{\longrightarrow}G'X^{n}$.
By Lemma~\ref{computation-lemma} this response is equivalent to any other response from Defender.
\item
If Attacker chooses the action $p_iX^n \act{a} q_1X^n$, Defender responds with $p_i'X^n \rightarrow G'X^n \rightarrow^{*} G'X^{3n} \act{a} q_1'X^{3n}$, making use of Lemma~\ref{computation-lemma}.
Attacker's optimal move is to choose $(q_1X^n, q_1'X^{3n})$ to be the next configuration.
\item
Now Attacker would not do a $t$ action since $t_{3}X^n \beq t_{1}X^{3n}$.
It chooses the action $a$ and the new configuration $(q_2X^n, q_2'X^{3n})$.
\item
Then we come to another semantic Defender's Forcing.
If Attacker plays $q_{2}X^{n}\stackrel{\tau}{\longrightarrow}GX^{n}$, Defender plays $q_{2}'X^{n}\stackrel{\tau}{\longrightarrow}GX^{3n}$; and vice versa.
\item
If Attacker chooses the transition $q_2'X^{3n} \act{a} q_3'X^{3n}$, Defender's response is $q_2X^n \stackrel{\tau}{\longrightarrow} GX^n \Longrightarrow  GX^{3n} \act{a} q_3X^{3n}$, exploiting again Lemma~\ref{computation-lemma}.
Attacker's nontrivial choice of the new configuration is $(q_3X^{3n}, q_3'X^{3n})$.
\item
Finally Attacker would not choose a $t_{1}$ action since $t_{1}X^{3n} \beq t_{1}X^{3n}$.
So after an $a$ action, the configuration becomes $(q_jX^{3n}, q_j'X^{3n})$.
\end{enumerate}
It is easy to see that the configuration $(q_jX^{3n}, q_j'X^{3n})$ is optimal for both Attacker and Defender.
If $q_jX^{3n}\simeq q_j'X^{3n}$ then Defender's Forcing described above is justified.
If $q_jX^{3n}\not\simeq q_j'X^{3n}$ the forcing is ineffective since Attacker can choose to play $p_{i}X^{n}\stackrel{\tau}{\longrightarrow}G'X^{n}$ and wins.
\qed\end{proof}

The main result of the section follows easily from Lemma~\ref{2014-01-19} and its proof.

\begin{theorem}\label{undecidability-of-PRS-with-state}
On $\mathrm{nOCN}$ every relation $\mathcal{R}$ satisfying $\simeq\;\subseteq\mathcal{R}\subseteq\;\approx$ is undecidable.
\end{theorem}
\begin{proof}
Dividing a number by a constant can be encoded in similar fashion.
The rest of Mayr's reduction does not refer to any silent transitions.
It follows that we can construct a reduction witnessing that ``$\mathcal{M}$ halts iff $p_1X \not\beq p_1'X$''.
As a matter of fact the reduction supports the stronger correspondence stated as follows: ``$\mathcal{M}$ halts iff $p_1X \not\approx p_1'X$''.
\qed\end{proof}

\section{Undecidability of nPA}\label{sec-Undecidability-of-nPA}

Following~\cite{Srba2003}, our main undecidability result is proved by reducing PCP (Post's Correspondence Problem) to the branching bisimilarity checking problem on nPA.
Suppose $\Sigma$ is a finite set of symbols and $\Sigma^{+}$ is the set of nonempty finite strings over $\Sigma$.
The size of $\Sigma$ is at least two.
PCP is defined as follows.
\nprob{Post's Correspondence Problem}{$\{ (u_1, v_1), (u_2, v_2) \dots (u_n, v_n) ~|~ u_i ,v_i
  \in \Sigma^{+} \}$}{Are there $i_1, i_2, \dots i_m\in \{1, 2, \dots,n \}$ with $m\geq 1$ such that $u_{i_1}u_{i_2}\dots u_{i_m} = v_{i_1}v_{i_2}\dots v_{i_m}$?}
We will fix a PCP instance $\textrm{INST}{=}\{ (u_1, v_1), (u_2, v_2) \dots (u_n, v_n) ~|~ u_i,v_i \in \Sigma^{+} \}$ in this section.
 Our task is to construct a normed process algebra $\mathcal{G} {=} (\mathcal{C}, \mathcal{A}, \Delta)$ containing two process constants $X,Y$ that render true the following equivalence.
\begin{equation}\label{eq:main}
\textrm{``INST has a solution''}\ \mathrm{iff}\ X \beq Y\ \mathrm{iff}\ X\weq Y.
\end{equation}
We will prove~(\ref{eq:main}) by validating the following statements:
\begin{itemize}
\item ``If INST has a solution then $X\beq Y$''. This is Lemma~\ref{lm:defender_ws} of Section~\ref{sec:generate_phase}.
\item ``If INST has no solution then $X\not\weq Y$''. This is Lemma~\ref{lm:attacker_ws} of Section~\ref{sec:generate_phase}.
\end{itemize}
As $X\simeq Y$ implies $X\approx Y$, the main theorem of the paper follows from~(\ref{eq:main}).
\begin{theorem}\label{thm:main}
On $\mathrm{nPA}$ every relation $\mathcal{R}$ satisfying $\simeq\;\subseteq\mathcal{R}\subseteq\;\approx$ is undecidable.
\end{theorem}

In the rest of the section, we firstly define $\mathcal{G}$, and then argue in several steps how the game based on $\mathcal{G}$ works in Defender's favour if INST has a solution.

\subsection{The nPA Game}\label{sec-Defining-the-Game}

The construction of $\mathcal{G}=(\mathcal{C},\mathcal{A},\Delta)$ from INST is based on Srba's reduction~\cite{Srba2003}.
Substantial amount of redesigning effort is necessary to make it work for the {\em branching} bisimilarity on the {\em normed} PA.
The set $\mathcal{A}$ of actions is defined by
\begin{eqnarray*}
\mathcal{A} &=& \Lambda \cup \mathcal{N} \cup \Sigma \cup \{\tau\},
\end{eqnarray*}
where $\Lambda = \{\lambda_U, \lambda_V, \lambda_D, \lambda_I, \lambda_S, \lambda_Z\}$, $\mathcal{N} = \{1, \dots, n\}$ and $\Sigma,n$ are from INST.
The set $\mathcal{C}$ of process constants is defined by
\begin{eqnarray*}
\mathcal{C} &=& \{ X, Y, Z, I, S, C, C', D, G, G', G_u, G_v, G_v' \} \cup \mathcal{U} \cup \mathcal{V} \cup \mathcal{W}, \\
\mathcal{U} &=& \{U_i \mid i \in \mathcal{N}\}, \\
\mathcal{V} &=& \{V_i \mid i \in \mathcal{N}\}, \\
\mathcal{W} &=& \{ W(\omega, i), W(\omega,0) \mid \omega\in (\mathcal{SF}(u_i)\cup \mathcal{SF}(v_i)) \ \mathrm{and}\  i \in \mathcal{N}\},
\end{eqnarray*}
where for each $\omega \in \Sigma^{*}$, the notation $\mathcal{SF}(\omega)$ stands for the set of suffixes of $\omega$.
The set of transition rules is given in Fig.~\ref{fig:delta}.
It is clear from these rules that $\mathcal{G}$ is indeed normed.
In particular $P\Longrightarrow\epsilon$ for all $P\in\mathcal{U}\cup\mathcal{V}\cup\mathcal{W}$.

\begin{figure*}[t]
\begin{center}
\begin{tabular}{|l|} \hline
$\ X \act{\lambda_U} D  {\,\|\,} G_v$,\ \ $X \stackrel{\tau}{\longrightarrow} D$;\ \ \ \ $Y \stackrel{\tau}{\longrightarrow} D$;\ \ \ \ $D \stackrel{\tau}{\longrightarrow} D {\,\|\,}  G_u$,\ \ $D \act{\lambda_D} C$; \\
$\ G_u\stackrel{\tau}{\longrightarrow} G_uU_i$,\ \ $G_u \act{\lambda_U} G_v U_i$; \ \ \ \ $G_u \stackrel{\tau}{\longrightarrow} G_v'$,\ \ $G_v' \stackrel{\tau}{\longrightarrow} G_v' V_i$,\ \ $G_v' \stackrel{\tau}{\longrightarrow} Z$; \\
\hline
$\ G_v \stackrel{\tau}{\longrightarrow}G_v V_i$,\ \ $G_v \stackrel{\tau}{\longrightarrow} \epsilon$,\ \ $G_v \act{\lambda_V} Z$;\ \ \ \ $Z\stackrel{\tau}{\longrightarrow} \epsilon$,\ \ $Z \act{\lambda_Z} \epsilon$; \\
\hline
$\ C \act{\lambda_I} I$,\ \ $C \act{\lambda_S} S$,\ \ $C \stackrel{\tau}{\longrightarrow} C  {\,\|\,}  G$,\ \ $C \stackrel{\tau}{\longrightarrow} C {\,\|\,}  G_v$; \\
$\ G \stackrel{\tau}{\longrightarrow} GU_i$,\ \ $G \stackrel{\tau}{\longrightarrow} GV_i$,\ \ $G \stackrel{\tau}{\longrightarrow} \epsilon$; \\
\hline
$\ I \act{\lambda_I} C'$,\ \ $I \act{i}I$;\ \ \ \ $S \act{\lambda_S} C'$,\ \ $S \act{a} S$;\ \ \ \ $C' \stackrel{\tau}{\longrightarrow} C'  {\,\|\,}  G'$,\ \ $C' \stackrel{\tau}{\longrightarrow} \epsilon$; \\
$\ G' \stackrel{\tau}{\longrightarrow} G'U_i$,\ \ $G' \stackrel{\tau}{\longrightarrow} G' V_i$,\ \ $G'\stackrel{\tau}{\longrightarrow} G'W$,\ \ $G'\stackrel{\tau}{\longrightarrow}G_{v}$,\ \ $G'\stackrel{\tau}{\longrightarrow} Z$; \\
\hline
$\ U_i \stackrel{\tau}{\longrightarrow} W( u_i, i)$,\ \ $V_i \stackrel{\tau}{\longrightarrow} W(v_i, i)$; \\
$\ W( a\omega, i) \act{a} W(\omega, i)$,\ \ $W( a\omega,0) \act{a} W(\omega,0)$,\ \ $W(\omega, i) \act{i} W(\omega,0)$,  \\
$\ W(a\omega, i) \stackrel{\tau}{\longrightarrow} W(\omega, i)$,\ \ $W(a\omega,0) \stackrel{\tau}{\longrightarrow} W(\omega,0)$,\ \ $W(\omega, i) \stackrel{\tau}{\longrightarrow} W(\omega,0)$,\ \ $W(\epsilon,0) \stackrel{\tau}{\longrightarrow} \epsilon$.$\ $ \\ \hline\hline
\ In the above rules, $i$ ranges over $\{1, \dots, n\}$, $a$ ranges over $\Sigma$, and $W$ ranges over $\mathcal{W}$. \\ \hline
\end{tabular}
\end{center}
\caption{Transition Rules for the nPA Game. \label{fig:delta}}
\end{figure*}

We write $\mathbb{P}_u$, respectively $\mathbb{P}_v$, for a sequential composition of members of $\mathcal{U}$, respectively $\mathcal{V}$.
Similarly we write $\mathbb{P}$, respectively $\mathbb{Q}$, for a sequential composition of members of $\mathcal{U}\cup\mathcal{V}$, respectively $\mathcal{U}\cup\mathcal{V}\cup\mathcal{W}$.
If for example the sequence $u$ is empty, $\mathbb{P}_u$ is understood to denote $\epsilon$.

\subsection{Defender's Generator}\label{sec-Generator-for-Defender}

To explain how the reduction works we start with the generators introduced by the process algebra.
A generator should be able to not only produce what is necessary but also do away with what has been produced.
The process $D$ for instance can induce circular silent transition sequence of the form
\[
D \stackrel{\tau}{\longrightarrow} D {\,\|\,} G_u \Longrightarrow  D {\,\|\,} G_u\mathbb{P}_u \stackrel{\tau}{\longrightarrow} D {\,\|\,} G_v'\mathbb{P}_u \Longrightarrow  D {\,\|\,} G_v'\mathbb{P}_v\mathbb{P}_u \Longrightarrow  D.
\]
By Lemma~\ref{computation-lemma} all the processes appearing in the above sequence are branching bisimilar.
Notice that the only reason the process constant $G_v'$ is introduced is to make available the above circular sequence.
The constant $G_v'$ is necessary because $G_u$ cannot reach $G_v$ via silent moves.
Similar circular silent transition sequences are also available for $C$ and $C'$.

\begin{lemma}\label{cor:gen}
Suppose $P \in \{D,C,C'\}$ and $P\Longrightarrow  P {\,\|\,} Q $.
Then $P {\,\|\,} Q\Longrightarrow P$.
\end{lemma}

\begin{corollary}\label{lm:forcing}
The following equalities are valid  for all $\mathbb{P}_u,\mathbb{P}_v,\mathbb{P},\mathbb{Q}$.
  \begin{enumerate}
  \item $D \beq D {\,\|\,} G_u\mathbb{P}_u \beq D {\,\|\,} G_{v}'\mathbb{P}_v\mathbb{P}_u  \beq D {\,\|\,} Z\mathbb{P}_v\mathbb{P}_u \beq D {\,\|\,} \mathbb{P}_v\mathbb{P}_u \beq D {\,\|\,}  W\mathbb{P}_v\mathbb{P}_u$;
  \item $C \beq C {\,\|\,} G\mathbb{ P} \beq C {\,\|\,} \mathbb{ P} \beq C {\,\|\,} W\mathbb{P} \beq C {\,\|\,} G_v\mathbb{P}_v$;
  \item $C' \beq C' {\,\|\,}  G'\mathbb{Q} \beq C' {\,\|\,}  G_{v}\mathbb{Q} \beq C' {\,\|\,} Z\mathbb{Q} \beq C' {\,\|\,} \mathbb{Q}$.
  \end{enumerate}
\end{corollary}

It has been observed that generating transitions are the most tricky ones in decidability proofs~\cite{Stirling2001,Czerwinski2011,Fu2013}.
Here they are used to Defender's advantage.
A generator can start everything all over again from scratch.
This gives Defender the ability to copy Attacker if the latter does not make a particular move.

The bisimulation game of $(X,Y)$ is played in two phases.
The generating phase comes first.
During this phase Defender tries to produce a pair $\mathbb{P}_u,\mathbb{P}_v$, via Defender's Forcing using the generators, that encode a solution to INST.
Next comes the checking phase in which Attacker tries to reject the pair $\mathbb{P}_u,\mathbb{P}_v$.
In the light of the delayed effect of Defender's Forcing in branching bisimulation games, we will look at the two phases in reverse order.

\subsection{Checking Phase}\label{sec:check_phase}

The processes $U_{i},V_{i}$ play two roles.
One is to announce $u_{i}$, respectively $v_{i}$; the other is to reveal the index $i$.
The first role can be suppressed by composing $U_{i}$, respectively $V_{i}$, with $S$ while the second can be discharged by composing with $I$~\cite{Srba2003}.
Since $I,S$ are normed, Attacker can choose to remove $I$, respectively $S$.
In our game the removal can be done by playing $I \act{\lambda_I} C'$, respectively $S\act{\lambda_S}C'$.
According to (3) of Corollary~\ref{lm:forcing} however Attacker would lose immediately if it plays $I \act{\lambda_I} C'$, respectively $S\act{\lambda_S}C'$, in a branching bisimulation game starting from $(I{\,\|\,}\mathbb{Q}, I{\,\|\,}\mathbb{Q}')$, respectively $(S{\,\|\,}\mathbb{Q}, S{\,\|\,}\mathbb{Q}')$.
Notice that it is important for a process constant $W$ to ignore the string/index information by doing silent transitions.
Otherwise the interleaving between actions in $\Sigma$ and actions in $\mathcal{N}$ would defeat Defender's attempt to prove string/index equality.

\begin{lemma}\label{pro:index_string}
  Suppose $\mathbb{U}= U_{i_1} U_{i_2} \dots U_{i_l}$, $\mathbb{V} = V_{j_1} V_{j_2} \dots V_{j_r}$ and $B\in\{\epsilon, Z, G_v\}$.
  The following statements are valid, where $\approxeq\;\in\{\simeq,\approx\}$.
  \begin{enumerate}
  \item $I{\,\|\,}B\mathbb{P}\mathbb{U} \approxeq I{\,\|\,}B\mathbb{P}\mathbb{V}$ if and only if $u_{i_1}u_{i_2}
    \dots u_{i_l} = v_{j_1} v_{j_2} \dots v_{j_r}$.
  \item $S{\,\|\,}B\mathbb{P}\mathbb{U} \approxeq S{\,\|\,}B\mathbb{P}\mathbb{V}$ if and only if $i_1i_2\ldots i_l = j_1j_2\ldots j_r$.
  \end{enumerate}
\end{lemma}
\begin{proof}
Suppose $I{\,\|\,}B\mathbb{P}\mathbb{U} \beq I{\,\|\,}B\mathbb{P}\mathbb{V}$ and w.l.o.g. $|u_{i_1}u_{i_2}\dots u_{i_l}|\ge|v_{j_1} v_{j_2} \dots v_{j_r}|$.
An action sequence from $I{\,\|\,}B\mathbb{P}\mathbb{U}$ to $I{\,\|\,}\mathbb{U}$ must be simulated essentially by an action sequence from $I{\,\|\,}B\mathbb{P}\mathbb{V}$ to $I{\,\|\,}\mathbb{V}$.
But then $u_{i_1}u_{i_2}\dots u_{i_l} = v_{j_1} v_{j_2} \dots v_{j_r}$ can be derived from $I{\,\|\,}\mathbb{U}\simeq I{\,\|\,}\mathbb{V}$.
The converse implication follows from the discussion in the above.
The second equivalence can be proved similarly.
\qed\end{proof}

The following proposition, in which $\approxeq\;\in\{\simeq,\approx\}$, says that the constant $C$ can be used to check both string equality and index equality by Attacker's forcing.

\begin{proposition}\label{pro:check}
If $\mathbb{U} = U_{i_1}U_{i_2} \dots U_{i_l}$ and $\mathbb{V} = V_{j_1} V_{j_2} \dots V_{j_r}$,
then for all $\mathbb{P}$, $C{\,\|\,}Z\mathbb{P}\mathbb{U} \approxeq C{\,\|\,}Z\mathbb{P}\mathbb{V}$ iff
$i_1i_2\ldots i_l = j_1j_2\ldots j_r$ and $u_{i_1}u_{i_2} \dots u_{i_l} = v_{j_1} v_{j_2} \dots v_{j_r}$.
\end{proposition}
\begin{proof}
In one direction we prove that $C{\,\|\,}Z\mathbb{P}\mathbb{U} \approx C{\,\|\,}Z\mathbb{P}\mathbb{V}$ implies
$i_1i_2\ldots i_l = j_1j_2\ldots j_r$ and $u_{i_1}u_{i_2} \dots u_{i_l} = v_{j_1} v_{j_2} \dots v_{j_r}$.
If $i_1i_2\ldots i_l \not= j_1j_2\ldots j_r$, then Attacker chooses $C{\,\|\,}Z\mathbb{P}\mathbb{U} \act{\lambda_S} S{\,\|\,}Z\mathbb{P}\mathbb{U}$.
  Defender cannot invoke the action $Z\stackrel{\tau}{\longrightarrow}\epsilon$ for otherwise an $\lambda_{Z}$ action cannot be performed before an $\lambda_{V}$ action.
  The process constant $Z$ is introduced precisely for this blocking effect.
  Defender's play must be of the form $C{\,\|\,}Z\mathbb{P}\mathbb{V} \Longrightarrow C{\,\|\,}Q{\,\|\,}Z\mathbb{P}\mathbb{V}
  \act{\lambda_S} S{\,\|\,}Q{\,\|\,}Z\mathbb{P}\mathbb{V} \Longrightarrow S{\,\|\,}Q'{\,\|\,}Z\mathbb{P}\mathbb{V}$.
  If $Q'$ can perform any one of $\{\lambda_{V},\lambda_{Z}\}\cup\mathcal{N}$, Attacker wins since $S$ can do none of those.
  If $Q'$ can do none of those actions, then $S \beq S {\,\|\,}  Q'$.
  By Lemma~\ref{pro:index_string} Attacker has a winning strategy for the weak bisimulation game $(S{\,\|\,}Z\mathbb{P}\mathbb{U}, S{\,\|\,}Q'{\,\|\,}Z\mathbb{P}\mathbb{V})$.
  If $u_{i_1}u_{i_2} \dots u_{i_l} \not= v_{j_1} v_{j_2} \dots v_{j_r}$, the argument is similar.

Conversely we prove that $i_1i_2\ldots i_l = j_1j_2\ldots j_r \wedge u_{i_1}u_{i_2} \dots u_{i_l} = v_{j_1} v_{j_2} \dots v_{j_r}$ implies $C{\,\|\,}Z\mathbb{P}\mathbb{U} \simeq C{\,\|\,}Z\mathbb{P}\mathbb{V}$.
This is done by showing that the relation
\[\left\{ (C{\,\|\,}Q{\,\|\,}Z\mathbb{P}\mathbb{U}, C{\,\|\,}Q{\,\|\,}Z\mathbb{P}\mathbb{V}) \left| \begin{array}{l}
     i_1i_2\ldots i_l = j_1j_2\ldots j_r \\
     u_{i_1}u_{i_2}\dots u_{i_l} = v_{j_1}v_{j_2}\dots v_{j_r}.
    \end{array}
\right.\right\}\ \cup \simeq\]
is a branching bisimulation.
\qed\end{proof}

\subsection{Generating Phase}\label{sec:generate_phase}

Suppose that INST has a solution $i_1,i_2,\dots,i_k$.
Fix the following abbreviations: $\mathbb{U}^{-}=U_{i_2}\dots U_{i_k}$, $\mathbb{U}=U_{i_1}\mathbb{U}^{-}$ and $\mathbb{V}=V_{i_1}V_{i_2}\dots V_{i_k}$.
We will argue that Defender has a winning strategy in the branching bisimulation game of $(X,Y)$.
Defender's basic idea is to produce the pair $\mathbb{U},\mathbb{V}$ by forcing.
Its strategy and Attacker's counter strategy are described below.
\begin{enumerate}[(i)]
\item
By Defender's Forcing Attacker plays $X\stackrel{\lambda_{U}}{\longrightarrow}D{\,\|\,}G_v$.
Defender proposes $\mathbb{U}$ via the transitions $Y\stackrel{\tau}{\longrightarrow}D\stackrel{\tau}{\longrightarrow}D{\,\|\,}G_u\Longrightarrow D{\,\|\,}G_u\mathbb{U}^{-}\stackrel{\lambda_{U}}{\longrightarrow}D{\,\|\,}G_v\mathbb{U}$.
The use of an explicit action $\lambda_{U}$ guarantees that $\mathbb{U}$ is {\em nonempty}.
Now Attacker has a number of configurations to choose from.
But by (1) of Corollary~\ref{lm:forcing}, it all boils down to choosing $(D{\,\|\,}G_v,D{\,\|\,}G_v\mathbb{U})$.
\item
Due to (1) of Corollary~\ref{lm:forcing} Attacker would not remove $G_v$ using either $G_v\stackrel{\tau}{\longrightarrow}\epsilon$ or $G_v\stackrel{\lambda_{V}}{\longrightarrow}Z$.
It can generate an element of $\mathcal{V}$ using $G_v$.
It can do an action induced by $D$ or a descendant of $D$.
Defender simply copycats Attacker's actions.
The configuration stays in the form $(D{\,\|\,}Q{\,\|\,}G_v\mathbb{P}_v,D{\,\|\,}Q{\,\|\,}G_v\mathbb{P}_v\mathbb{U})$.
\item
To have any chance to win, Attacker must try the action $\lambda_{D}$.
Defender does the same action.
The configuration becomes $(C{\,\|\,}Q{\,\|\,}G_v\mathbb{P}_v,C{\,\|\,}Q{\,\|\,}G_v\mathbb{P}_v\mathbb{U})$.
At this point if Attacker plays a harmless action, Defender can copycat the action; and the configuration stays in the same shape.
\item
An important observation is that if Attacker plays $C{\,\|\,}Q{\,\|\,}G_v\mathbb{P}_v\stackrel{\ell}{\longrightarrow}P_{1}$, Defender can play $C{\,\|\,}Q{\,\|\,}G_v\mathbb{P}_v\mathbb{U}\Longrightarrow C{\,\|\,}Q\Longrightarrow C{\,\|\,}Q{\,\|\,}G_v\mathbb{P}_v\stackrel{\ell}{\longrightarrow}P_{1}$ and wins.
Here $C{\,\|\,}Q \simeq C{\,\|\,}Q{\,\|\,}G_v\mathbb{P}_v$ by (2) of Corollary~\ref{lm:forcing}.
To see that the assumptions $i_1i_2\ldots i_l = j_1j_2\ldots j_r$ and $u_{i_1}u_{i_2} \dots u_{i_l} = v_{j_1} v_{j_2} \dots v_{j_r}$ imply $C{\,\|\,}Q{\,\|\,}G_v\mathbb{P}_v\mathbb{U} \simeq C{\,\|\,}Q$, notice that $C{\,\|\,}Q{\,\|\,}G_v\mathbb{P}_v\mathbb{U}\Longrightarrow C{\,\|\,}Q \Longrightarrow C{\,\|\,}Q{\,\|\,}G_v\mathbb{P}_v\mathbb{V}$ and that $C{\,\|\,}Q{\,\|\,}G_v\mathbb{P}_v\mathbb{U} \simeq C{\,\|\,}Q{\,\|\,}G_v\mathbb{P}_v\mathbb{V}$ is a corollary of Proposition~\ref{pro:check}.
Thus Attacker would choose $C{\,\|\,}Q{\,\|\,}G_v\mathbb{P}_v\mathbb{U}$ to continue.
\item
Attacker would not play $C{\,\|\,}Q{\,\|\,}G_v\mathbb{P}_v\mathbb{U} \stackrel{\tau}{\longrightarrow} C{\,\|\,}Q{\,\|\,}\mathbb{P}_v\mathbb{U}$ because it would lose right away according to  (2) of Corollary~\ref{lm:forcing}.
\item
By Lemma~\ref{pro:index_string} Attacker would not do a $\lambda_{I}$ action or a $\lambda_{S}$ action.
It stands the best chance to play $C{\,\|\,}Q{\,\|\,}G_v\mathbb{P}_v\mathbb{U} \stackrel{\lambda_{V}}{\longrightarrow} C{\,\|\,}Q{\,\|\,}Z\mathbb{P}_v\mathbb{U}$.
The counter play from Defender is $C{\,\|\,}Q{\,\|\,}G_v\mathbb{P}_v\Longrightarrow C{\,\|\,}Q{\,\|\,}G_v\mathbb{P}_v\mathbb{V} \stackrel{\lambda_{V}}{\longrightarrow} C{\,\|\,}Q{\,\|\,}Z\mathbb{P}_v\mathbb{V}$.
\end{enumerate}
The last configuration $(C{\,\|\,}Q{\,\|\,}Z\mathbb{P}_v\mathbb{V}, C{\,\|\,}Q{\,\|\,}Z\mathbb{P}_v\mathbb{U})$ is optimal for Attacker.
By Proposition~\ref{pro:check} Defender has a winning strategy for the branching bisimulation game of $(C{\,\|\,}Q{\,\|\,}Z\mathbb{P}_v\mathbb{V}, C{\,\|\,}Q{\,\|\,}Z\mathbb{P}_v\mathbb{U})$.
Hence the following lemma.

\begin{lemma}\label{lm:defender_ws}
If INST has a solution then $X\beq Y$.
\end{lemma}

The converse of Lemma~\ref{lm:defender_ws} also holds.
In fact a stronger result is obtainable.
In the weak bisimulation game of $(X,Y)$, Attacker has a strategy to force the game to reach a configuration that is essentially of the form $(C{\,\|\,}Z\mathbb{P}_v', C{\,\|\,}Z\mathbb{P}_v\mathbb{P}_u)$, where $\mathbb{P}_u \neq \epsilon$.
If there is no solution to INST, Proposition~\ref{pro:check} implies $C{\,\|\,}Z\mathbb{P}_v'\not\approx C{\,\|\,}Z\mathbb{P}_v\mathbb{P}_u$.
It follows that Attacker has a winning strategy for the weak bisimulation game of $(X,Y)$.

\begin{lemma}\label{lm:attacker_ws}
If INST has no solution then $X\not\weq Y$.
\end{lemma}

\section{Conclusion}\label{sec-Conclusion}

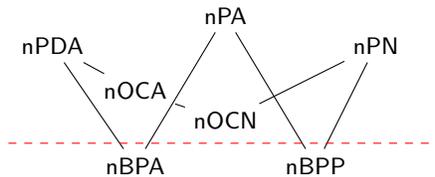
\begin{figure*}[t]
\begin{center}
\tikzstyle{rec} = [draw, thin, text=red]
    \begin{tikzpicture}
         \node (bpa) at (-1.2,0) {$\mathsf{nBPA}$};
         \node (bpp) at (1.2,0) {$\mathsf{nBPP}$};
         \node (ocn) at (0,0.6) {$\mathsf{nOCN}$};
         \node (oca) at (-1.2,1) {$\mathsf{nOCA}$};
         \node (pa) at (0,2) {$\mathsf{nPA}$};
         \node (pda) at (-2.3,1.6) {$\mathsf{nPDA}$};
         \node (pn) at (2,1.6) {$\mathsf{nPN}$};
         \node (a) at (-3, 0.3){};
         \node (b) at (3, 0.3){};
          \path[-]
          (ocn) edge (oca)
          (ocn) edge (pn)
          (oca) edge (pda)
          (bpa) edge (pda)
          (bpa) edge (pa)
          (bpp) edge (pa)
          (bpp) edge (pn)
          ;
          \path[red][dashed]
          (a) edge (b);
    \end{tikzpicture}
\end{center}
\caption{Decidability Border for Branching Bisimilarity on Normed PRS}
\label{fig:beq_border}
\end{figure*}

Putting together the results derived in this paper, we see that there is a decidability border in the normed PRS hierarchy, see Fig.~\ref{fig:beq_border}.
The branching bisimilarity
\begin{enumerate}
\item
is undecidable on all normed models above either nBPA or nBPP, and
\item
is decidable for both nBPP and nBPA~\cite{Czerwinski2011,Fu2013}.
\end{enumerate}
We have confirmed that the first statement is valid for the weak bisimilarity, which slightly strengthens the results obtained in~\cite{KuceraJancar2006}.
In fact the statement is valid for every relation between the branching bisimilarity and the weak bisimilarity.
It has been conjectured that the second statement is also true for the weak bisimilarity.
The answers however have remained a secret for us up to now.

Tighter complexity bounds, or even completeness characterizations, would be very welcome.
Another avenue for further study is based on the observation that although the undecidability results of both the present paper and the paper of Jan\v{c}ar and Srba~\cite{JancarSrba2008} are about the same models, the degrees of undecidability are most likely to be different.
In~\cite{JancarSrba2008} it is pointed out that by constraining the silent actions of nPDA, say to $\epsilon$-popping or $\epsilon$-pushing silent moves, the degree of undecidability of the weak bisimilarity goes from the analytic hierarchy down to the arithmetic hierarchy.
It is therefore a reasonable hope that the same restriction may lead to decidable results for the branching bisimilarity on some PRS models.
Further studies are called for.

An extended abstract of this paper has been accepted for publication~\cite{YinFuHeHuangTao2014}. 

\vspace*{5mm}\noindent {\bf Acknowledgement}.
We gratefully acknowledge the support of the National Science Foundation of China (61033002, ANR 61261130589, 91318301).
We thank the anonymous referees and Patrick Totzke for their constructive suggestions.



\appendix

\section{Proof of Corollary~\ref{lm:forcing}}\label{A-forcing}

The proof is a simple application of Lemma~\ref{computation-lemma}.
For the constant $D$ one has the following circular silent transition sequence:
\begin{eqnarray*}
D &\stackrel{\tau}{\longrightarrow}& D{\,\|\,}G_u  \\
 &\Longrightarrow& D{\,\|\,}G_u\mathbb{P}_{u} \\
 &\stackrel{\tau}{\longrightarrow}& D{\,\|\,}G_v'\mathbb{P}_{u} \\
 &\Longrightarrow& D{\,\|\,}G_v'\mathbb{P}_{v}\mathbb{P}_{u} \\
 &\stackrel{\tau}{\longrightarrow}& D{\,\|\,}Z\mathbb{P}_{v}\mathbb{P}_{u} \\
 &\stackrel{\tau}{\longrightarrow}& D{\,\|\,}\mathbb{P}_{v}\mathbb{P}_{u} \\
 &\stackrel{\tau}{\longrightarrow}& D{\,\|\,}W\mathbb{P}_{v}\mathbb{P}_{u} \\
 &\Longrightarrow& D{\,\|\,}\mathbb{P}_{u} \\
 &\Longrightarrow& D{\,\|\,}W\mathbb{P}_{u} \\
 &\Longrightarrow& D.
\end{eqnarray*}
For the constant $C$ one has
\begin{eqnarray*}
C &\stackrel{\tau}{\longrightarrow}& C{\,\|\,}G \\
 &\Longrightarrow& C{\,\|\,}G\mathbb{P} \\
 &\Longrightarrow& C{\,\|\,}\mathbb{P} \\
 &\Longrightarrow& C{\,\|\,}W\mathbb{P} \\
 &\Longrightarrow& C \\
 &\stackrel{\tau}{\longrightarrow}& C{\,\|\,}G_v \\
 &\Longrightarrow& C{\,\|\,}G_v\mathbb{P}_v \\
 &\stackrel{\tau}{\longrightarrow}& C{\,\|\,}\mathbb{P}_v \\
 &\Longrightarrow& C.
\end{eqnarray*}
Finally for the constant $C'$ one has
\begin{eqnarray*}
C' &\stackrel{\tau}{\longrightarrow}& C'{\,\|\,}G' \\
 &\Longrightarrow& C'{\,\|\,}G'\mathbb{Q} \\
 &\stackrel{\tau}{\longrightarrow}& C'{\,\|\,}Z\mathbb{Q} \\
 &\Longrightarrow& C' \\
 &\Longrightarrow& C'{\,\|\,}G'\mathbb{Q} \\
 &\stackrel{\tau}{\longrightarrow}& C'{\,\|\,}G_v\mathbb{Q} \\
 &\Longrightarrow& C'.
\end{eqnarray*}
We are done.

\section{Proof of Lemma~\ref{pro:index_string}}\label{A-index-string}

Suppose $\mathbb{U}= U_{i_1} U_{i_2} \dots U_{i_l}$ and $\mathbb{V} = V_{j_1} V_{j_2} \dots V_{j_r}$.
We show that
  \begin{enumerate}[(i)]
  \item If $u_{i_1}u_{i_2}\dots u_{i_l} = v_{j_1} v_{j_2} \dots v_{j_r}$ then $I{\,\|\,}G_v\mathbb{P}\mathbb{U} \simeq I{\,\|\,}G_v\mathbb{P}\mathbb{V}$.
  \item If $i_1i_2\ldots i_l = j_1j_2\ldots j_r$ then $S{\,\|\,}G_v\mathbb{P}\mathbb{U} \simeq S{\,\|\,}G_v\mathbb{P}\mathbb{V}$.
  \item If $I{\,\|\,}G_v\mathbb{P}\mathbb{U} \approx I{\,\|\,}G_v\mathbb{P}\mathbb{V}$ then $u_{i_1}u_{i_2}\dots u_{i_l} = v_{j_1} v_{j_2} \dots v_{j_r}$.
  \item If $S{\,\|\,}G_v\mathbb{P}\mathbb{U} \approx S{\,\|\,}G_v\mathbb{P}\mathbb{V}$ then $i_1i_2\ldots i_l = j_1j_2\ldots j_r$.
  \end{enumerate}

\begin{proof}
(ii) Suppose $i_1i_2\ldots i_l = j_1j_2\ldots j_r$.
The proof is given by the following case analysis:
\begin{itemize}
\item
If Attacker chooses the transition $S{\,\|\,}G_v\mathbb{P}\mathbb{U}\stackrel{\lambda_{S}}{\longrightarrow}C'{\,\|\,}G_v\mathbb{P}\mathbb{U}$, Defender can win by playing $S{\,\|\,}G_v\mathbb{P}\mathbb{V}\stackrel{\lambda_{S}}{\longrightarrow}C'{\,\|\,}G_v\mathbb{P}\mathbb{V}$.
This is because $C'{\,\|\,}G_v\mathbb{P}\mathbb{U}\simeq C' \simeq C'{\,\|\,}G_v\mathbb{P}\mathbb{V}$ by (3) of Corollary~\ref{lm:forcing}.
\item If Attacker plays a transition caused by an action of $G_v$, Defender does the same action.
Suppose the resulting configuration is $(S{\,\|\,}Z\mathbb{P}\mathbb{U},S{\,\|\,}Z\mathbb{P}\mathbb{V})$.
Attacker would not play a $\lambda_{S}$ action for the same reason.
If it plays an action caused by $Z$, Defender follows suit.

\item
If both Attacker and Defender play in the optimal manner, the game will reach the configuration $(S{\,\|\,}\mathbb{U}, S{\,\|\,}\mathbb{V})$.
By (3) of Corollary~\ref{lm:forcing} Attacker would lose if it plays a $\lambda_{S}$ action.
It would not win if it plays an action from $\Sigma$.
Finally if Attacker decides to play say $i_{1}$ or skip it, Defender copycats the action.
By the assumption $i_1i_2\ldots i_l = j_1j_2\ldots j_r$, Attacker would not win in this case either.
\end{itemize}
This completes the proof of (ii).

(iv) Suppose that $S{\,\|\,}G_v\mathbb{P}\mathbb{U} \approx S{\,\|\,}G_v\mathbb{P}\mathbb{V}$ and without loss of generality that
\begin{equation}\label{2014-01-28}
|i_1i_2\ldots i_l| \ge |j_1j_2\ldots j_r|.
\end{equation}
Now $S{\,\|\,}G_v\mathbb{P}\mathbb{V} \stackrel{\lambda_{V}}{\longrightarrow}
S{\,\|\,}Z\mathbb{P}\mathbb{V}$ must be bisimulated by $S{\,\|\,}G_v\mathbb{P}\mathbb{U}
\stackrel{\lambda_{V}}{\Longrightarrow} S{\,\|\,}Z\mathbb{P}'\mathbb{P}\mathbb{U}$ for some $\mathbb{P}'$.
Notice that if $\mathbb{P}'$ is not empty, there would be no hope that $S{\,\|\,}Z\mathbb{P}'\mathbb{P}\mathbb{U}\approx S{\,\|\,}Z\mathbb{P}\mathbb{V}$.
So the simulation must be of the form $S{\,\|\,}G_v\mathbb{P}\mathbb{U} \stackrel{\lambda_{V}}{\longrightarrow} S{\,\|\,}Z\mathbb{P}\mathbb{U}$.
Let \[S{\,\|\,}Z\mathbb{P}\mathbb{U} \stackrel{\lambda_Z}{\longrightarrow}S{\,\|\,}\mathbb{P}\mathbb{U} \stackrel{k_{1}}{\Longrightarrow}\ldots\stackrel{k_{m}}{\Longrightarrow} S{\,\|\,}\mathbb{U}\]
be the longest sequence of actions of the form $\lambda_Z,k_{1},\ldots,k_{m}$ such that $k_{1},\ldots,k_{m}\in\mathcal{N}$.
In the light of (\ref{2014-01-28}) the simulation from $\mathbb{P}\mathbb{V}$ must be of the form
\[S{\,\|\,}Z\mathbb{P}\mathbb{V}\stackrel{\lambda_Z}{\longrightarrow}S{\,\|\,}\mathbb{P}\mathbb{V} \stackrel{k_{1}}{\Longrightarrow}\ldots\stackrel{k_{m}}{\Longrightarrow} S{\,\|\,}\mathbb{V}.\]
By similar argument one shows that $S{\,\|\,}\mathbb{U}\approx S{\,\|\,}\mathbb{V}$ implies $i_1i_2\ldots i_l = j_1j_2\ldots j_r$.

The proof of (i) and (iii) can be done in the same fashion.
\qed\end{proof}

\section{Proof of Proposition~\ref{pro:check}}\label{A-check}

Suppose $\mathbb{U} = U_{i_1}U_{i_2} \dots U_{i_l}$ and $\mathbb{V} = V_{j_1} V_{j_2} \dots V_{j_r}$.
We have seen that $i_1i_2\ldots i_l = j_1j_2\ldots j_r$ and $u_{i_1}u_{i_2} \dots u_{i_l} = v_{j_1} v_{j_2} \dots v_{j_r}$ imply $C{\,\|\,}Z\mathbb{P}\mathbb{U} \simeq C{\,\|\,}Z\mathbb{P}\mathbb{V}$ for all $\mathbb{P}$.
The following lemma says that if there is some $\mathbb{P}$ such that $C{\,\|\,}Z\mathbb{P}\mathbb{U} \simeq C{\,\|\,}Z\mathbb{P}\mathbb{V}$, then $i_1i_2\ldots i_l = j_1j_2\ldots j_r$ and $u_{i_1}u_{i_2} \dots u_{i_l} = v_{j_1} v_{j_2} \dots v_{j_r}$.

\begin{lemma}\label{weak-check-variant}
If $C{\,\|\,}Z\mathbb{P}_v^1 \approx C{\,\|\,}Z\mathbb{P}_v^{2}\mathbb{P}_u^{2}$ for some $\mathbb{P}_v^1,\mathbb{P}_v^{2},\mathbb{P}_u^{2}$ with $\mathbb{P}_u^{2}\ne\epsilon$, then {\em INST} has a solution.
\end{lemma}
\begin{proof}
Suppose $C{\,\|\,}Z\mathbb{P}_v'\stackrel{\lambda_{I}}{\longrightarrow}I{\,\|\,}Z\mathbb{P}_v'$ is simulated by
\[C{\,\|\,}Z\mathbb{P}_v^{2}\mathbb{P}_u^{2} \Longrightarrow C{\,\|\,}Q{\,\|\,}Z\mathbb{P}_v^{2}\mathbb{P}_u^{2} \stackrel{\lambda_{I}}{\longrightarrow} I{\,\|\,}Q{\,\|\,}Z\mathbb{P}_v^{2}\mathbb{P}_u^{2} \Longrightarrow I{\,\|\,}Q'{\,\|\,}Z\mathbb{P}_v^{2}\mathbb{P}_u^{2}\approx I{\,\|\,}Z\mathbb{P}_v^1\]
for some $Q,Q'$.
It is easy to see that $Q'$ contains neither $G$ nor $G_v$.
Moreover $Q'$ contains no processes of the form $W(\omega,i)$ for $\omega\ne\epsilon$.
The only nontrivial components $Q'$ may contain are processes of the form $W(\epsilon,i)$.
It follows that $I{\,\|\,}Q'\approx I$.
Consequently $I{\,\|\,}Z\mathbb{P}_v^1 \approx I{\,\|\,}Z\mathbb{P}_v^{2}\mathbb{P}_u^{2}$.
Using similar argument one derives that $S{\,\|\,}Z\mathbb{P}_v^1 \approx S{\,\|\,}Z\mathbb{P}_v^{2}\mathbb{P}_u^{2}$.
We are done by applying Lemma~\ref{pro:index_string}.
\qed\end{proof}

\section{Proof of Lemma~\ref{lm:defender_ws}}\label{A-defender-ws}

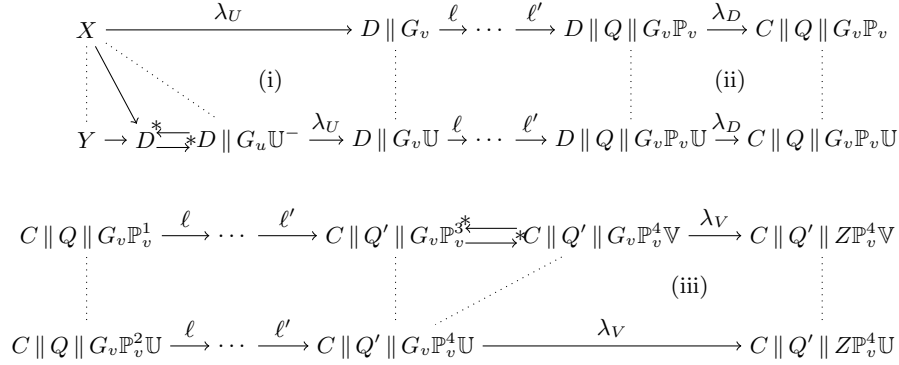
\begin{figure*}[t]
\begin{center}
\resizebox{\textwidth}{!}{
\tikzstyle{rec} = [draw, thin, text=red]
    \begin{tikzpicture}
      \node(aa) at (2.5,-0.7) {(i)};
      \node(aa) at (8.75,-0.7) {(ii)};
         \node(a1) at (0,0) {$X$};
        \node (b1) at (0,-1.5) {$Y$};
        \node(c1) at (0.8,-1.5) {$D$};
        \node(c2) at (2.2,-1.5) {$D{\,\|\,}G_u\mathbb{U}^{-}$};
        \node(a2) at (4.2,0) {$D{\,\|\,}G_v$};
        \node(b2) at (4.2,-1.5) {$D{\,\|\,}G_v\mathbb{U}$};
        \path[ ->] (a1) edge[above] node{$\lambda_U$} (a2);
        \path[->] (c2) edge[above] node{$\lambda_U$} (b2);
        \path[->] (a1) edge (c1);
        \path[->] (b1) edge (c1);
        \path[<-]  (0.95,-1.4) edge (1.40,-1.4);
        \path[->]  (0.95,-1.60) edge (1.40,-1.60);
        \node (bs1) at (1.42,-1.50) {$*$};
        \node (as1) at (0.95,-1.3) {$*$};
        \node(a3) at (7.4,0) {$D{\,\|\,}Q{\,\|\,}G_v\mathbb{P}_v$};
        \node(b3) at (7.4,-1.5) {$D{\,\|\,}Q{\,\|\,}G_v\mathbb{P}_v\mathbb{U}$};
        \node(a4) at (10,0) {$C{\,\|\,}Q{\,\|\,}G_v\mathbb{P}_v$};
        \node(b4) at (10,-1.5) {$C{\,\|\,}Q{\,\|\,}G_v\mathbb{P}_v\mathbb{U}$};

        \path[->] (a3) edge[above] node{$\lambda_D$} (a4);
        \path[->] (b3) edge[above] node{$\lambda_D$} (b4);
        \node(a23) at (5.5,0) {$\dots$};
        \node(b23) at (5.5,-1.5) {$\dots$};
        \path[->] (a2) edge[above] node{$\ell$} (a23);
        \path[->] (b2) edge[above] node{$\ell$} (b23);
        \path[->] (a23) edge[above] node{$\ell'$} (a3);
        \path[->] (b23) edge[above] node{$\ell'$} (b3);

        \path[dotted] (a1) edge (c2);
        \path[dotted] (a1) edge (b1);
        \path[dotted] (a2) edge (b2);
        \path[dotted] (a3) edge (b3);
        \path[dotted] (a4) edge (b4);

        \node(aa) at (8.2,-3.5) {(iii)};

        \node(3a1) at (0,-2.8) {$C{\,\|\,}Q{\,\|\,}G_v\mathbb{P}_v^{1}$};
        \node (3b1) at (0,-4.3) {$C{\,\|\,}Q{\,\|\,}G_v\mathbb{P}_v^{2}\mathbb{U}$};

        \node(3a12) at (2,-2.8) {$\dots$};
        \node(3b12) at (2,-4.3) {$\dots$};
        \path[->] (3a1) edge[above] node{$\ell$} (3a12);
        \path[->] (3b1) edge[above] node{$\ell$} (3b12);
        \node(3a2) at (4.2,-2.8) {$C{\,\|\,}Q'{\,\|\,}G_v\mathbb{P}_v^{3}$};
        \node(3b2) at (4.2,-4.3) {$C{\,\|\,}Q'{\,\|\,}G_v\mathbb{P}_v^{4}\mathbb{U}$};
        \path[->] (3a12) edge[above] node{$\ell'$} (3a2);
        \path[->] (3b12) edge[above] node{$\ell'$} (3b2);

        \node (3a3) at (10,-2.8) {$C{\,\|\,}Q'{\,\|\,}Z\mathbb{P}_v^{4}\mathbb{V}$};
        \node (3b3) at (10,-4.3) {$C{\,\|\,}Q'{\,\|\,}Z\mathbb{P}_v^{4}\mathbb{U}$};
        \node (3a23) at (7,-2.8)  {$C{\,\|\,}Q'{\,\|\,}G_v\mathbb{P}_v^{4}\mathbb{V}$};
        \path[->] (3b2) edge[above] node{$\lambda_V$}(3b3);
        \path[->] (3a23) edge[above] node{$\lambda_V$}(3a3);

        \path[<-] (5.15,-2.7) edge (5.85,-2.7);
        \path[->] (5.15,-2.9) edge (5.85,-2.9);
        \node(3as1) at (5.9,-2.8) {$*$};
        \node(3bs1) at (5.15,-2.6)  {$*$};
        \path[dotted] (3a3) edge (3b3);
        \path[dotted] (3a23) edge (3b2);
        \path[dotted] (3a1) edge (3b1);
        \path[dotted] (3a2) edge (3b2);
    \end{tikzpicture}}
\end{center}
\caption{Defender's Strategy}
\label{fig:defender_ws}
\end{figure*}

Defender's strategy is composed of three substrategies (see Fig.~\ref{fig:defender_ws}).
We now give the details of the substrategies.
\begin{enumerate}[(i)]
\item
By Defender's Forcing with delayed justification, Attacker chooses to play $X\act{\lambda_U}D{\,\|\,}G_v$.
Defender responds with the following transition sequence
\[ Y \stackrel{\tau}{\longrightarrow}D \Longrightarrow^{*}  D{\,\|\,}G_u\mathbb{U}^{-} \act{\lambda_U} D{\,\|\,}G_v\mathbb{U},\]
noticing that $Y \beq D \beq D{\,\|\,}G_u\mathbb{U}^{-}$ according to (1) of Corollary~\ref{lm:forcing}.
  The following case analysis implies that if Attacker plays optimal, it would continue from the configuration $(D{\,\|\,}G_v, D{\,\|\,}G_v\mathbb{U})$.
  \begin{enumerate}[(a)]
  \item If Attacker sets the configuration to be $(D{\,\|\,}G_v, D{\,\|\,}G_v\mathbb{U})$, we are done.
  \item Otherwise assume w.l.o.g. that Attacker sets it to be $(X, D{\,\|\,}G_v\mathbb{U}^{-})$.
  \begin{itemize}
  \item
  By Defender's Forcing, Attacker would not play $D{\,\|\,}G_v\mathbb{U}^{-} \act{\ell} Q$ since it can be matched by $X\stackrel{\tau}{\longrightarrow} D \Longrightarrow D{\,\|\,}G_v\mathbb{U}^{-} \act{\ell} Q$.
  \item
  Attacker would not play $X\stackrel{\tau}{\longrightarrow} D$ either since Defender can win by palying $D{\,\|\,}G_u\mathbb{U}^{-}\Longrightarrow D$.
  \item
  If Attacker plays $X\act{\lambda_U}D{\,\|\,}G_v$, Defender responds with the transition $D{\,\|\,}G_u\mathbb{U}^{-} \act{\lambda_U} D{\,\|\,}G_v\mathbb{U}$.
  \end{itemize}
  \end{enumerate}

\item
Now suppose the current configuration is $(D{\,\|\,}Q{\,\|\,}G_v\mathbb{P}_v, D{\,\|\,}Q{\,\|\,}G_v\mathbb{P}_v\mathbb{U})$.
Attacker would choose neither $G_v \stackrel{\tau}{\longrightarrow} \epsilon$ nor $G_v \act{\lambda_V} Z$ since $D{\,\|\,}Q{\,\|\,}\mathbb{P}_v \beq D{\,\|\,}Q{\,\|\,}\mathbb{P}_v\mathbb{U}$ and $D{\,\|\,}Q{\,\|\,}Z\mathbb{P}_v \beq D{\,\|\,}Q{\,\|\,}Z\mathbb{P}_v\mathbb{U}$ by (1) of Corollary~\ref{lm:forcing}.
The other cases are as follows:
\begin{itemize}
\item
Attacker plays $D{\,\|\,}Q{\,\|\,}G_v\mathbb{P}_v\act{\ell}D{\,\|\,}Q'{\,\|\,}G_v\mathbb{P}_v$.
Defender responds by playing $D{\,\|\,}Q{\,\|\,}G_v\mathbb{P}_v\mathbb{U}\act{\ell}D{\,\|\,}Q'{\,\|\,}G_v\mathbb{P}_v\mathbb{U}$.
\item
Attacker plays $D{\,\|\,}Q{\,\|\,}G_v\mathbb{P}_v\act{\tau}D{\,\|\,}G_v{\,\|\,}Q{\,\|\,}G_v\mathbb{P}_v$.
Defender responds by playing $D{\,\|\,}Q{\,\|\,}G_v\mathbb{P}_v\mathbb{U}\act{\tau}D{\,\|\,}G_v{\,\|\,}Q{\,\|\,}G_v\mathbb{P}_v\mathbb{U}$.
\item
Attacker plays $D{\,\|\,}Q{\,\|\,}G_v \mathbb{P}_v\act{\tau}D{\,\|\,}Q{\,\|\,}G_v V_i\mathbb{P}_v$.
Defender responds by playing $D{\,\|\,}Q{\,\|\,}G_v \mathbb{P}_v\mathbb{U}\act{\tau}D{\,\|\,}Q{\,\|\,}G_v V_i\mathbb{P}_v\mathbb{U}$.
\item
Attacker plays $D{\,\|\,}Q{\,\|\,}G_v\mathbb{P}_v\act{\lambda_D}C{\,\|\,}Q{\,\|\,}G_v\mathbb{P}_v$.
Defender counter plays $D{\,\|\,}Q{\,\|\,}G_v\mathbb{P}_v\mathbb{U} \act{\lambda_D}C{\,\|\,}Q{\,\|\,}G_v\mathbb{P}_v\mathbb{U}$.
In this case Attacker can only choose $(C{\,\|\,}Q{\,\|\,}G_v\mathbb{P}_v ,C{\,\|\,}Q{\,\|\,}G_v\mathbb{P}_v\mathbb{U})$ as the next configuration.
\end{itemize}
In these cases Attacker will eventually choose to play an $\lambda_D$ action to have any chance to win at all.

If Attacker chooses $D{\,\|\,}Q{\,\|\,}G_v\mathbb{P}_v\mathbb{U}$ to play, the situations are symmetric.

\item
For generality suppose $(C{\,\|\,}Q{\,\|\,}G_v\mathbb{P}_v^{1}, C{\,\|\,}Q{\,\|\,}G_v\mathbb{P}_v^{2}\mathbb{U})$ is the current configuration.
Attacker would not choose any transition of the form
\[C{\,\|\,}Q{\,\|\,}G_v\mathbb{P}_v^{1} \act{\ell} P_1\]
since the following response
\begin{equation}\label{2014-01-27}
C{\,\|\,}Q{\,\|\,}G_v\mathbb{P}_v^{2}\mathbb{U}\Longrightarrow  C{\,\|\,}Q \Longrightarrow  C{\,\|\,}Q{\,\|\,}G_v\mathbb{P}_v^{1} \act{\ell} P_1
\end{equation}
is a winning move for Defender.
To see that none of the silent transitions appearing in~(\ref{2014-01-27}) are change-of-state, first notice that $C{\,\|\,}Q \simeq C{\,\|\,}Q{\,\|\,}G_v\mathbb{P}_v^{1}$ by (2) of Corollary~\ref{lm:forcing}.
The equivalence $C{\,\|\,}Q{\,\|\,}G_v\mathbb{P}_v^{2}\mathbb{U} \simeq C{\,\|\,}Q$ is derived as follows:
One has that
\[C{\,\|\,}Q{\,\|\,}G_v\mathbb{P}_v^{2}\mathbb{U}\Longrightarrow  C{\,\|\,}Q \Longrightarrow C{\,\|\,}Q{\,\|\,}G_v\mathbb{P}_v^{2}\mathbb{V}.\]
It is easy to see that Proposition~\ref{pro:check} implies $C{\,\|\,}G_v\mathbb{P}_v^{2}\mathbb{U} \simeq C{\,\|\,}G_v\mathbb{P}_v^{2}\mathbb{V}$.
It then follows from $C{\,\|\,}Q{\,\|\,}G_v\mathbb{P}_v^{2}\mathbb{U} \simeq C{\,\|\,}Q{\,\|\,}G_v\mathbb{P}_v^{2}\mathbb{V}$ and Lemma~\ref{computation-lemma} that $C{\,\|\,}Q{\,\|\,}G_v\mathbb{P}_v^{2}\mathbb{U} \simeq C{\,\|\,}Q$.

Now suppose Attacker chooses $C{\,\|\,}Q{\,\|\,}G_v\mathbb{P}_v^{2}\mathbb{U}$ to play.
\begin{enumerate}[(a)]
\item
If Attacker plays some $C{\,\|\,}Q{\,\|\,}G_v\mathbb{P}_v^{2}\mathbb{U} \act{\ell}P_{2}$ caused by either an action of $Q$ or $C\stackrel{\tau}{\longrightarrow}C{\,\|\,}G$ or $C\stackrel{\tau}{\longrightarrow}C{\,\|\,}G_{v}$, Defender plays the same action, reaching to a configuration of the same shape.
\item
If Attacker plays $C{\,\|\,}Q{\,\|\,}G_v\mathbb{P}_v^{2}\mathbb{U}\act{\lambda_I} I{\,\|\,}Q{\,\|\,}G_v\mathbb{P}_v^{2}\mathbb{U}$, Defender replies
\[C{\,\|\,}Q{\,\|\,}G_v\mathbb{P}_v^{1} \Longrightarrow  C{\,\|\,}Q \Longrightarrow C{\,\|\,}Q{\,\|\,}G_v\mathbb{P}_v^{2}\mathbb{V} \act{\lambda_I} I{\,\|\,}Q{\,\|\,}G_v\mathbb{P}_v^{2}\mathbb{V}.\]
Suppose Attacker chooses $(I{\,\|\,}Q{\,\|\,}G_v\mathbb{P}_v^{2}\mathbb{U},I{\,\|\,}Q{\,\|\,}G_v\mathbb{P}_v^{2}\mathbb{V})$ to be the next configuration.
Now $C{\,\|\,}Q{\,\|\,}G_v\mathbb{P}_v^{1} \simeq C{\,\|\,}Q{\,\|\,}G_v\mathbb{P}_v^{2}\mathbb{V}$ by (2) of Corollary~\ref{lm:forcing} and $I{\,\|\,}Z\mathbb{P}_v^{2}\mathbb{V} \simeq I{\,\|\,}Z\mathbb{P}_v^{2}\mathbb{U}$ by Lemma~\ref{pro:index_string}.
It follows that \[I{\,\|\,}Q{\,\|\,}G_v\mathbb{P}_v^{2}\mathbb{V} \simeq I{\,\|\,}Q{\,\|\,}G_v\mathbb{P}_v^{2}\mathbb{U}.\]
So in this case Defender wins.
\item
The situation is similar if Attacker chooses to play an $\lambda_S$ action.
\end{enumerate}
Attacker will not win if it keeps doing (a).
Eventually it must do (b) or (c).
\end{enumerate}
By applying the three substrategies consecutively we see that Attacker's optimal strategy is to reach a configuration of the form $(C{\,\|\,}Q{\,\|\,}G_v\mathbb{P}_v\mathbb{U},C{\,\|\,}Q{\,\|\,}G_v\mathbb{P}_v\mathbb{V})$.
This is however also a win situation for Defender because of the equivalence
\[C{\,\|\,}G_v\mathbb{P}_v\mathbb{U} \simeq C{\,\|\,}G_v\mathbb{P}_v\mathbb{V},\]
the simple proof of which is as follows:
\begin{itemize}
\item If $C{\,\|\,}G_v\mathbb{P}_v\mathbb{U}$ performs a $\lambda_{I}$ or $\lambda_{S}$ action, then $C{\,\|\,}G_v\mathbb{P}_v\mathbb{V}$ does the same.
We are done by Lemma~\ref{pro:index_string}.
\item If $C{\,\|\,}G_v\mathbb{P}_v\mathbb{U}$ does an action induced by $G_v\stackrel{\tau}{\longrightarrow}\epsilon$, the process $C{\,\|\,}G_v\mathbb{P}_v\mathbb{V}$ follows suit.
We are done by Corollary~\ref{lm:forcing}.
\item If $C{\,\|\,}G_v\mathbb{P}_v\mathbb{U}$ acts using $G_v\stackrel{\lambda_{V}}{\longrightarrow}Z$, then $C{\,\|\,}G_v\mathbb{P}_v\mathbb{V}$ copycats the action.
We are done by Proposition~\ref{pro:check}.
\item If $C{\,\|\,}G_v\mathbb{P}_v\mathbb{U}\stackrel{\tau}{\longrightarrow}C{\,\|\,}G_v V_i\mathbb{P}_v\mathbb{U}$, then $C{\,\|\,}G_v\mathbb{P}_v\mathbb{V}\stackrel{\tau}{\longrightarrow}C{\,\|\,}G_v V_i\mathbb{P}_v\mathbb{V}$.
    We get a pair of processes of the same shape.
\end{itemize}
This completes the proof.

\section{Proof of Lemma~\ref{lm:attacker_ws}}\label{A-attacker-ws}

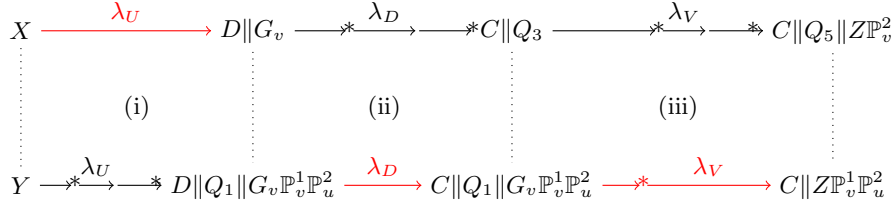
\begin{figure*}[t]
\begin{center}
\resizebox{\textwidth}{!}{
\tikzstyle{rec} = [draw, thin, text=red]
    \begin{tikzpicture}
        \node(a1) at (0,0) {$X$};
        \node (b1) at (0,-2) {$Y$};
        \path[dotted] (a1) edge (b1);

        \node(a2) at (3,0) {$D\|G_v$};
        \node(b2) at (3,-2) {$D\|Q_1\|G_v\mathbb{P}^{1}_v\mathbb{P}_u^{2}$};
        \path[red, ->] (a1) edge[above] node{$\lambda_U$} (a2);
        \node(bs1) at (0.7,-1.95) {$*$};
        \node(bs2) at (1.75,-1.95){$*$};
        \path[->] (b1) edge (0.65,-2);
        \path[->] (0.75,-2) edge[above] node{$\lambda_U$} (1.2,-2);
        \path[->] (1.25,-2) edge (b2);
        \path[dotted] (a2) edge (b2);

        \node(a3) at (6.35,0) {$C\|Q_3$};
        \node(b3) at (6.35,-2) {$C\|Q_1\|G_v\mathbb{P}^{1}_v\mathbb{P}_u^{2}$};
        \path[dotted] (a3) edge (b3);
        \path[red, ->] (b2) edge[above] node{$\lambda_D$} (b3);
        \path[->] (a2) edge (4.2,0);
        \path[->] (4.3,0) edge[above] node{$\lambda_D$} (5.1,0);
        \path[->] (5.15,0) edge (a3);
        \node(as1) at (4.25,0.05) {$*$};
        \node(as2) at (5.85,0.05){$*$};

        \node(b4) at (9.3,-2) {};
        \node(as3) at (8.25,0.05) {$*$};
        \node[red](bs3) at (8.05,-1.95) {$*$};

        \node(a5) at (10.5,0) {$C\|Q_5\|Z\mathbb{P}^2_v$};
        \node(b5) at (10.5,-2) {$C\|Z\mathbb{P}^{1}_v\mathbb{P}_u^{2}$};
        \path[dotted] (a5) edge (b5);
         \path[red,->] (8.1,-2) edge [above] node{$\lambda_V$}(b5);
        \path[red,->] (b3) edge (8,-2);
        \path[->] (a3) edge (8.2,0);
        \path[->] (8.3,0) edge[above] node{$\lambda_V$} (8.85,0);
        \path[->] (8.9,0) edge (a5);
        \node(as4) at (9.47,0.05) {$*$};
        \node(i) at (1.5,-1) {(\romannumeral1)};
        \node(ii) at (4.7,-1) {(\romannumeral2)};
        \node(iii) at (8.5,-1) {(\romannumeral3)};
    \end{tikzpicture}}
\end{center}
\caption{Attacker's Strategy}
\label{fig:weak_attacker_ws}
\end{figure*}

Attacker's winning strategy is very much similar to the optimal strategy described in Section~\ref{A-defender-ws}.
It is outlined in Fig.~\ref{fig:weak_attacker_ws}, where Attacker's moves are marked in red.
We explain the strategy in the following.
\begin{itemize}
\item
Attacker plays $X\act{\lambda_U}D{\,\|\,}G_v$.
Defender's optimal response would be
\[Y \Longrightarrow  D{\,\|\,}G_u\mathbb{P}_u^{1} \act{\lambda_U} D{\,\|\,}G_v\mathbb{P}_u^{2} \Longrightarrow D {\,\|\,} Q_1{\,\|\,}G_v\mathbb{P}_v^{1}\mathbb{P}_u^{2} \]
 for some $Q_1$, $\mathbb{P}_v^{1}$, $\mathbb{P}_u^{1}$ and $\mathbb{P}_u^{2}$ such that $D\Longrightarrow  D{\,\|\,}Q_1$ and $\mathbb{P}_u^{2}=U_i\mathbb{P}_u^{1}$ for some  $i\in \mathcal{N}$.
Defender would not remove $G_v$ using the rule $G_v\stackrel{\tau}{\longrightarrow}\epsilon$ because that would make it unable to reply Attacker's next move $D{\,\|\,}G_v\stackrel{\lambda_{V}}{\longrightarrow}D{\,\|\,}Z$.
\item
Attacker then plays $D {\,\|\,} Q_1{\,\|\,}G_v\mathbb{P}_v^{1}\mathbb{P}_u^{2} \act{\lambda_D} C {\,\|\,} Q_1{\,\|\,}G_v\mathbb{P}_v^{1}\mathbb{P}_u^{2}$.
Defender's response must be of the form $D{\,\|\,}G_v \Longrightarrow  D{\,\|\,} Q_2 \act{\lambda_D} C{\,\|\,} Q_2 \Longrightarrow  C{\,\|\,}Q_3$ for some $Q_2,Q_3$.
\item
It is easy to see that $Q_1\Longrightarrow\epsilon$.
So the following is a valid move of Attacker:
\[C{\,\|\,}Q_1{\,\|\,}G_v\mathbb{P}_v^{1}\mathbb{P}_u^{2} \Longrightarrow C{\,\|\,}G_v\mathbb{P}_v^{1}\mathbb{P}_u^{2} \act{\lambda_V} C{\,\|\,}Z\mathbb{P}_v^{1}\mathbb{P}_u^{2}.\]
Defender's response must be a transition sequence of the following form
\[C{\,\|\,}Q_3 \Longrightarrow  C{\,\|\,}Q_4{\,\|\,}G_v\mathbb{P}_v^{2} \act{\lambda_V}C{\,\|\,}Q_4{\,\|\,}Z\mathbb{P}_v^{2} \Longrightarrow  C{\,\|\,}Q_5{\,\|\,}Z\mathbb{P}_v^{2}\]
for some $\mathbb{P}_v^{2}$, $Q_4$ and $Q_5$.
Now $Q_5$ must be a parallel composition of processes that can be generated by $C$ or $D$.
W.l.o.g. assume that \[Q_5=Q^{C}{\,\|\,}Q^{D}_1{\,\|\,}\dots{\,\|\,}Q^{D}_k\] where $Q^{C}$ is generated by $C$ and $D^{D}_i$ is generated by $D$ for $i\in\{1,\ldots,k\}$.
There are following cases:
\begin{enumerate}[(a)]
\item
If $Q_{i}$ contains an occurrence of $G_u$, then Attacker wins because the process $C{\,\|\,}Q_5{\,\|\,}Z\mathbb{P}_v^{2}$ can do a $\lambda_{U}$ action that cannot be simulated by the process $C{\,\|\,}Z\mathbb{P}_v^{1}\mathbb{P}_u^{2}$.
\item
If $Q_{i}$ contains an occurrence of $Z$, then Attacker also wins since the process $C{\,\|\,}Q_5{\,\|\,}Z\mathbb{P}_v^{2}$ can do two consecutive $\lambda_{Z}$ actions whereas the process $C{\,\|\,}Z\mathbb{P}_v^{1}\mathbb{P}_u^{2}$ can do only one such action.
\item
If for each $i\in\{1,\ldots,k\}$ the process $Q_{i}$ contains neither $G_u$ nor $Z$ then $C\beq C{\,\|\,}Q^{D}_i$ for all $i$.
By Lemma~\ref{cor:gen} we must also have $C{\,\|\,}Q^{C}\simeq C$.
It follows that $C{\,\|\,}Q_{5}\simeq C$.
So in this case the configuration the game reached is essentially $(C{\,\|\,}Z\mathbb{P}_v^{2},C{\,\|\,}Z\mathbb{P}_v^{1}\mathbb{P}_u^{2})$.
By Lemma~\ref{weak-check-variant}, $C{\,\|\,}Z\mathbb{P}_v^{2} \not\approx C{\,\|\,}Z\mathbb{P}_v^{1}\mathbb{P}_u^{2}$, meaning that Attacker has a winning strategy.
\end{enumerate}
\end{itemize}
We are done.

\end{document}

%% file: preamble-llncs.tex
\usepackage{amssymb}
\usepackage{amsmath}
\usepackage{enumerate}
\usepackage{graphicx}
\usepackage{semantic}
\usepackage{mathtools}
\usepackage{proof}
\usepackage[colorlinks, citecolor=red]{hyperref}
\usepackage[all]{xy}

\newcommand{\act}[2][]{{\stackrel{#2}{\longrightarrow}}^{#1}}

\newcommand{\beq}{\simeq}
\newcommand{\weq}{\approx}

\usepackage{soul}
\setstcolor{red}
\newlength{\probwidth}
\newcommand{\nprob}[4][8]{
\begin{center}\normalfont\fbox{
\addtolength{\probwidth}{#1cm}\parbox{\probwidth}{\textsc{#2}\\\hspace*{1.5em}
\begin{tabular}[t]{rp{#1cm}}
\textit{Input:}&#3.\\\textit{Problem:}&#4
\end{tabular}}}
\end{center}}